\documentclass[english,journal,twoside,web]{ieeecolor}
\usepackage[T1]{fontenc}
\usepackage[latin9]{inputenc}
\usepackage{array}
\usepackage{multirow}
\usepackage{amsmath}
\usepackage{amssymb}
\usepackage{graphicx}

\makeatletter

\providecommand{\tabularnewline}{\\}

\usepackage{generic}
\usepackage{balance}
\usepackage{dsfont}

\usepackage{cite}
\usepackage{times}   

\newtheorem{theorem}{Theorem}[section]

\newtheorem{remark}[theorem]{Remark}
\newtheorem{corollary}[theorem]{Corollary}
\newtheorem{lemma}[theorem]{Lemma}
\newtheorem{proposition}[theorem]{Proposition}

\renewcommand{\QED}{\QEDopen}
\usepackage[all]{xy}

\usepackage{balance}
\usepackage{psfrag}

\usepackage{tikz}
\usetikzlibrary{positioning, arrows, matrix}

\usepackage{psfrag}

\usepackage{dsfont}

\usepackage{url}

\usepackage{tikz}
\usepackage{pgfplots}

\usepackage{url}
\usepackage{verbatim}
\newenvironment{inarxiv}{
}{
}

\newenvironment{inpaper}{
\comment
}{
\endcomment
}

\makeatother

\usepackage{babel}
\begin{document}
\title{Instabilizability Conditions for Continuous-Time Stochastic Systems Under Control Input Constraints} 
\author{Ahmet Cetinkaya, \IEEEmembership{Member, IEEE}, Masako Kishida, \IEEEmembership{Senior Member, IEEE}
\thanks{This work was supported by JST ERATO HASUO Metamathematics for Systems Design Project (No.\ JPMJER1603) and by JSPS KAKENHI Grant Number 20K14771.}
\thanks{Ahmet Cetinkaya and Masako Kishida are with the National Institute of Informatics, Tokyo, 101-8430, Japan.  (e-mail: {\mbox cetinkaya@nii.ac.jp, kishida@nii.ac.jp}).}}

\maketitle
\thispagestyle{empty}

\begin{abstract}In this paper, we investigate constrained control
of continuous-time linear stochastic systems. We show that for certain
system parameter settings, constrained control policies can never
achieve stabilization. Specifically, we explore a class of control
policies that are constrained to have a bounded average second moment
for Ito-type stochastic differential equations with additive and multiplicative
noise. We prove that in certain settings of the system parameters
and the bounding constant of the control constraint, divergence of
the second moment of the system state is inevitable regardless of
the initial state value and regardless of how the control policy is
designed. \end{abstract} 

\begin{IEEEkeywords} Stochastic systems, constrained control, linear systems\end{IEEEkeywords}

\section{Introduction}

\IEEEPARstart{S}{tabilization} under control input constraints is
an important research problem due to its wide applicability to systems
with actuator saturation. The works \cite{tarbouriech2011stability,book-saberi2012internal}
describe the challenges of this problem and provide comprehensive
discussions of the important results. A key result on this problem
is an impossibility result: linear deterministic systems with strictly
unstable system matrices cannot be globally stabilized if the norm
of the control input is constrained to stay below a constant threshold
\cite{necessary-condition-null-controllability-sontag1984algebraic,sussmann1994}. 

There is a rapidly growing interest in exploring control input constraints
for stochastic systems. For instance, \cite{stochastic-mpc-hokayem2012stochastic,stochastic-average-constraint-korda2014stochastic,hewing2019scenario,mark2020}
proposed stochastic model predictive controllers with control constraints;
\cite{pereira2019learning} and \cite{wang2020reinforcement} developed
reinforcement learning control frameworks with constraints, \cite{chang2021}
investigated fuzzy controllers for stochastic systems with actuator
saturation. Constrained control of nonlinear stochastic systems was
investigated by \cite{ying2006stochastically} and \cite{stochastic-design-min2018output},
and moreover, \cite{stochastic-networked-mishra2018sparse} explored
control constraints in stochastic networked control systems. 

The work \cite{chatterjee2012onmeansquare} presented an impossibility
result for constrained control of discrete-time stochastic systems.
It was shown there that if the control input of a strictly unstable
discrete-time stochastic system is subject to hard norm-constraints,
then the second moment of the state always diverges under nonvanishing
and unbounded additive stochastic process noise. A common approach
to overcome the difficulties in the stabilization of strictly unstable
systems is to consider probabilistic constraints instead of hard deterministic
constraints. However, it was shown in \cite{cetinkaya2020TAC} that
under certain conditions, stabilization of a discrete-time linear
stochastic system is impossible even under probabilistic constraints. 

The scope of the impossibility results provided in the abovementioned
articles covers discrete-time stochastic systems with additive noise.
In this paper, we are motivated to expand this scope by addressing
two issues. First, we want to know if similar impossibility results
can be obtained for \emph{continuous-time} stochastic systems. Secondly,
we want to investigate the effects of both \emph{additive} and \emph{multiplicative}
noise terms. Handling multiplicative noise terms is important, since
such terms can characterize parametric uncertainties in the system
(see \cite{elghaoui1995,khasminskii2011stochastic}). As our main
contribution, we identify the scenarios where stabilization of a \emph{continuous-time
stochastic system (with both additive and multiplicative noise)} is
\emph{not possible} under \emph{probabilistically-constrained control
policies}. Specifically, we consider control policies that have bounded
time-averaged second moments. This class of control policies encapsulate
many types of controllers with (probabilistic or deterministic) control
constraints. We obtain conditions on the bounding value of the control
constraint, under which the second moment of the state \emph{diverges}
regardless of the controller choice and regardless of the initial
state value. 

Our analysis for the continuous-time systems with additive and multiplicative
noise has a few key differences from that for the discrete-time case
with additive-only noise provided in \cite{chatterjee2012onmeansquare,cetinkaya2020TAC}.
First, in our case, we handle Ito-type stochastic differential equations
with state-dependent noise terms characterizing multiplicative Wiener
noise. In addition, we develop a form of reverse Gronwall's inequality
to obtain lower bounds on functions with superlinear growth. Through
our analysis, we observe that combination of additive and multiplicative
noise can make systems harder to stabilize. Even systems that have
Hurwitz-stable system matrices can be impossible to stabilize with
constrained controllers under the combination of additive and multiplicative
noise. 

We organize the rest of the paper as follows. In Section~\ref{sec:Constrained-Control-of},
we describe the constrained control problem. Then in Sections~\ref{sec:Conditions-for-Impossibility}
and \ref{sec:Special-Setting-with}, we provide our impossibility
results for constrained control of continuous-time stochastic systems.
Finally, we present numerical examples in Section~\ref{sec:Numerical-Example}
and conclude the paper in Section~\ref{sec:Conclusion}. 

\textbf{Notation:} We denote the Euclidean norm by $\|\cdot\|$,
the trace operator by $\mathrm{tr}(\cdot)$, and the maximum eigenvalue
of a Hermitian matrix $H\in\mathbb{C}^{n\times n}$ by $\lambda_{\max}(H)$.
We use $H^{\frac{1}{2}}$ to represent the unique nonnegative-definite
Hermitian square root of a nonnegative-definite Hermitian matrix $H\in\mathbb{C}^{n\times n}$,
satisfying $H^{\frac{1}{2}}H^{\frac{1}{2}}=H$ and $(H^{\frac{1}{2}})^{*}=H^{\frac{1}{2}}$.
The identity matrix in $\mathbb{R}^{n\times n}$ is denoted by $I_{n}$.
The notations $\mathrm{\mathbb{P}}[\cdot]$ and $\mathbb{E}[\cdot]$
respectively denote the probability and expectation on a probability
space $(\Omega,\mathcal{F},\mathbb{P})$ with sample space $\Omega$
and $\sigma$-algebra $\mathcal{F}$. We consider a continuous-time
filtration $\{\mathcal{F}_{t}\}_{t\geq0}$ with $\mathcal{F}_{t_{1}}\subseteq\mathcal{F}_{t_{2}}\subseteq\mathcal{F}$
for $t_{1}\leq t_{2}$. Throughout the paper $\{W(t)=[W_{1}(t),\ldots,W_{\ell}(t)]^{\mathrm{T}}\in\mathbb{R}^{\ell}\}_{t\geq0}$
denotes the Wiener process. Here, for every $i\in\{1,\ldots,\ell\}$,
the process $\{W_{i}(t)\in\mathbb{R}\}_{t\geq0}$ is $\mathcal{F}_{t}$-adapted;
$\{W_{i}(t)\in\mathbb{R}\}_{t\geq0}$, $i\in\{1,\ldots,\ell\}$, are
independent processes. Moreover, $\overline{c}$ denotes the complex
conjugate of a complex number $c\in\mathbb{C}$, and $\mathrm{Re}(c)$
denotes its real part. We use $C^{*}$ to denote the complex conjugate
transpose of a complex matrix $C\in\mathbb{C}^{n\times m}$, that
is, $(C^{*}){}_{i,j}=\overline{C_{j,i}}$, $i\in\{1,\ldots,m\}$,
$j\in\{1,\ldots,n\}$. Given a vector $v\in\mathbb{R}^{n}$, and indices
$i,j\in\{1,\ldots,n\}$, $i\leq j$, we define $v_{i:j}\in\mathbb{R}^{j-i+1}$
as $v_{i:j}\triangleq[v_{i},\ldots,v_{j}]^{\mathrm{T}}$. 

\section{Constrained Control of Continuous-Time Linear Stochastic Systems}

\label{sec:Constrained-Control-of}

Consider the continuous-time linear stochastic system described by
the Ito-type stochastic differential equation 
\begin{align}
\mathrm{d}x(t) & =(Ax(t)+Bu(t))\mathrm{d}t+\left[\Psi(x(t)),\,D\right]\mathrm{d}W(t),\label{eq:continuous-time-system}
\end{align}
for $t\geq0$, where $x(t)\in\mathbb{R}^{n}$ is the state with deterministic
initial value $x(0)=x_{0}$, $u(t)\in\mathbb{R}^{m}$ is the control
input, and moreover, $\{W(t)\in\mathbb{R}^{\ell}\}_{t\geq0}$ is the
Wiener process. 

The matrices $A\in\mathbb{R}^{n\times n}$ and $B\in\mathbb{R}^{n\times m}$
are called system and input matrices, respectively. Moreover, $\Psi(x(t))\in\mathbb{R}^{n\times\ell_{1}}$
and $D\in\mathbb{R}^{n\times\ell_{2}}$ (with $\ell_{1}+\ell_{2}=\ell$)
are noise matrices. The matrix-valued function $\Psi\colon\mathbb{R}^{n}\to\mathbb{R}^{n\times\ell_{1}}$
characterizes the effects of multiplicative noise and it is given
by
\begin{align}
\Psi(x) & =[C_{1}x,C_{2}x,\ldots,C_{\ell_{1}}x],\label{eq:Psi-C-def}
\end{align}
where $C_{i}\in\mathbb{R}^{n\times n}$, $i\in\{1,\ldots,\ell_{1}\}$.
The matrix $D\in\mathbb{R}^{n\times\ell_{2}}$ in \eqref{eq:continuous-time-system}
is used for characterizing the effects of additive noise. 

Notice that $W_{1:\ell_{1}}(\cdot)$ enters in the dynamics as multiplicative
noise and $W_{\ell_{1}+1:\ell}(\cdot)$ enters as additive noise,
since $\left[\Psi(x(t)),\,D\right]\mathrm{d}W(t)=\Psi(x(t))\mathrm{d}W_{1:\ell_{1}}(t)+D\mathrm{d}W_{\ell_{1}+1:\ell}(t)$. 

In this paper, we are interested in a stabilization problem. Since
the Wiener process $W_{\ell_{1}+1:\ell}(\cdot)$ enters in the dynamics
in an additive way, the state and its moments cannot converge to $0$
regardless of the control input, unless $D=0$. For this reason, asymptotic
stabilization is not possible and a weaker notion of stabilization
is needed. In this paper, we consider the\emph{ bounded second-moment
stabilization} notion, where the control goal is to achieve $\sup_{t\geq0}\mathbb{E}[\|x(t)\|^{2}]<\infty$. 

We consider a stochastic constraint such that the time-averaged $2$nd
moment of $u(t)$ is bounded by $\hat{u}\geq0$, i.e., 
\begin{align}
\frac{1}{t}\int_{0}^{t}\mathbb{E}[\left\Vert u(\tau)\right\Vert ^{2}]\mathrm{d}\tau & \leq\hat{u},\quad t\geq0.\label{eq:average-u-bound-1}
\end{align}
This constraint is a relaxation of other types of control constraints,
i.e., the satisfaction of \eqref{eq:average-u-bound-1} does not necessarily
imply satisfaction of other constraints. Note on the other hand that
norm-constraints (e.g., $\|u(t)\|\leq\overline{u}$ or $\|u(t)\|_{\infty}\leq\overline{u}$),
time-averaged norm constraints (e.g., $\frac{1}{t}\int_{0}^{t}\|u(\tau)\|\mathrm{d}\mathrm{\tau}\leq\overline{u}$),
as well as first- and second-moment constraints (e.g., $\mathbb{E}[\|u(t)\|]\leq\overline{u}$
or $\mathbb{E}[\|u(t)\|^{2}]\leq\overline{u}$) all satisfy \eqref{eq:average-u-bound-1}
for certain values of $\hat{u}$. 

\begin{remark}The structure of \eqref{eq:average-u-bound-1} is motivated
by the networked control problem of a plant with a remotely located
controller. In this problem, control commands $u_{\mathrm{C}}(t)$
transmitted from the controller are subject to packet losses, and
the plant sets its input $u(t)$ to $0$ if there is a packet loss
(and to $u_{\mathrm{C}}(t)$ otherwise). The actuator at the plant
side has a hard constraint requiring $\|u_{\mathrm{C}}(t)\|^{2}\leq\overline{u}_{\mathrm{C}}$
for $t\geq0$. With randomness involved in packet losses, the plant
input $u(t)$ actually satisfies \eqref{eq:average-u-bound-1} with
$\hat{u}<\overline{u}_{\mathrm{C}}$ (see Section~IV.D of \cite{cetinkaya2020TAC}
for the specific form of $\hat{u}$). Even though the actuator may
be powerful ($\overline{u}_{\mathrm{C}}$ is large), unstable noisy
plants in certain scenarios cannot be stabilized if there are very
frequent packet losses, because in such cases $\hat{u}$ is much smaller
than $\overline{u}_{\mathrm{C}}$, and the controller is unable to
provide inputs with sufficient average energy to the plant. \hfill $\triangleleft$\end{remark} 

For given $A,B,\Psi,D$, our goal is to find a threshold for $\hat{u}$,
below which stabilization of \eqref{eq:continuous-time-system} becomes
impossible and the second moment $\mathbb{E}[\|x(t)\|^{2}]$ diverges
regardless of the controller design. 

The following lemmas are used in the derivation of our main result
in Section~\ref{sec:Conditions-for-Impossibility}. The first lemma
is related to the bounding value $\hat{u}$ of the control constraint
\eqref{eq:average-u-bound-1}. The second lemma is an extension of
Gronwall's lemma (see, e.g., \cite{king2003differential}), where
the key condition involves a linear term and the result provides a
lower bound instead of an upper bound. 

\begin{lemma} \label{Lemma-HatU} Let $\hat{u}\in[0,\infty)$, $\kappa\in(0,\infty)$
be scalars that satisfy $\hat{u}<\kappa$. Then $\mathcal{Q}\triangleq\left\{ q>1\colon\hat{u}<\kappa/q\right\} $
is non-empty. \end{lemma} \begin{proof} Let $\tilde{q}\triangleq2\kappa/(\hat{u}+\kappa)$.
Since $\hat{u}<\kappa$, we have $2\kappa>\hat{u}+\kappa$, which
implies $\tilde{q}>1$. Moreover, since $\hat{u}<\kappa$, we have
$\kappa/\tilde{q}=(\hat{u}+\kappa)/2>\hat{u}$. As both $\tilde{q}>1$
and $\hat{u}<\kappa/\tilde{q}$ hold, we have $\tilde{q}\in Q$, implying
that $\mathcal{Q}\neq\emptyset$.\end{proof} 

\begin{lemma} \label{LemmaInverseGronwall} Given scalars $c_{0},c_{1}\in\mathbb{R}$
and $\phi>0$, suppose that $h\colon[0,\infty)\to\mathbb{R}$ is a
continuous function that satisfies 
\begin{align}
h(t) & \geq c_{0}+c_{1}t+\phi\int_{0}^{t}h(\tau)\mathrm{d}\tau,\quad t\geq0.\label{eq:f-ineq}
\end{align}
 Then we have 
\begin{align}
h(t) & \geq c_{0}e^{\phi t}+(c_{1}/\phi)(e^{\phi t}-1),\quad t\geq0.\label{eq:inverse-gronwall}
\end{align}
 Moreover, if \eqref{eq:f-ineq} holds with equality, then \eqref{eq:inverse-gronwall}
holds with equality. \end{lemma} \begin{proof} Let $g(s)\triangleq c_{0}+c_{1}s+\phi\int_{0}^{s}h(\tau)\mathrm{d}\tau$
for $s\geq0$. Since $h$ is a continuous function, by fundamental
theorem of calculus, we have $\frac{\mathrm{d}g(s)}{\mathrm{d}s}=c_{1}+\phi h(s)$.
Note that \eqref{eq:f-ineq} implies $h(s)\geq g(s)$. Furthermore,
since $\phi>0$, 
\begin{align}
 & \frac{\mathrm{d}(g(s)e^{-\phi s})}{\mathrm{d}s}=\frac{\mathrm{d}g(s)}{\mathrm{d}s}e^{-\phi s}-\phi e^{-\phi s}g(s)\nonumber \\
 & \quad=c_{1}e^{-\phi s}+\phi e^{-\phi s}(h(s)-g(s))\geq c_{1}e^{-\phi s}.\label{eq:gineq}
\end{align}
By integrating both left- and far right-hand sides of the inequality
\eqref{eq:gineq} over the interval $[0,t]$, we get $g(t)e^{-\phi t}-g(0)\geq\frac{c_{1}}{\phi}(1-e^{-\phi t})$.
By using this inequality and $g(0)=c_{0}$, we obtain $g(t)\geq c_{0}e^{\phi t}+\frac{c_{1}}{\phi}(e^{\phi t}-1)$
for $t\geq0$, which implies \eqref{eq:inverse-gronwall}, since $h(t)\geq g(t)$.
Finally, if \eqref{eq:f-ineq} holds with equality, then $h(s)=g(s)$
for $s\geq0$, and thus, \eqref{eq:gineq} holds with equality, which
implies that \eqref{eq:inverse-gronwall} holds with equality. \end{proof} 

\section{Conditions for Impossibility of Stabilization}

\label{sec:Conditions-for-Impossibility}

In this section, we present our main result, which provides conditions
on the control constraint \eqref{eq:average-u-bound-1}, under which
the stochastic system \eqref{eq:continuous-time-system} is impossible
to be stabilized.

\begin{theorem} \label{TheoremContinuousTime} Consider the stochastic
system \eqref{eq:continuous-time-system}. Assume that there exist
a nonnegative-definite Hermitian matrix $R\in\mathbb{C}^{n\times n}\setminus\{0\}$
and a scalar $\phi_{\mathrm{L}}>0$ such that 
\begin{align}
A^{\mathrm{T}}R+RA+\sum_{i=1}^{\ell_{1}}C_{i}^{\mathrm{T}}RC_{i} & \geq\phi_{\mathrm{L}}R,\label{eq:A-cond-1}\\
\mathrm{tr}\left(D^{\mathrm{T}}RD\right) & >0.\label{eq:nonzero-noise-1}
\end{align}
If the control policy is $\mathcal{F}_{t}$-adapted and satisfies
\eqref{eq:average-u-bound-1} with 
\begin{align}
\hat{u} & <\begin{cases}
\phi_{\mathrm{L}}\mathrm{tr}\left(D^{\mathrm{T}}RD\right)/\beta_{\mathrm{U}}, & \mathrm{if}\,\,\beta_{\mathrm{U}}\neq0,\\
\infty, & \mathrm{otherwise},
\end{cases}\label{eq:ubar-cond-1}
\end{align}
 where $\beta_{\mathrm{U}}\triangleq\lambda_{\max}(B^{\mathrm{T}}RB)$,
then the second moment of the state diverges, that is,
\begin{align}
\lim_{t\to\infty}\mathbb{E}[\left\Vert x(t)\right\Vert ^{2}] & =\infty,\label{eq:divergence-1}
\end{align}
for any initial state $x_{0}\in\mathbb{R}^{n}$.\end{theorem}

\begin{proof}Let $V(x)\triangleq x^{\mathrm{T}}Rx$. As a first step,
we will show $\lim_{t\to\infty}\mathbb{E}\left[V(x(t))\right]=\infty.$
Let $\Lambda\triangleq A^{\mathrm{T}}R+RA+\sum_{i=1}^{\ell_{1}}C_{i}^{\mathrm{T}}RC_{i}$.
It follows from Ito formula (see Section~4.2 of \cite{oksendal2005stochastic})
that 
\begin{align}
\mathrm{d}V(x(t)) & =\mathrm{tr}(D^{\mathrm{T}}RD)\mathrm{d}t+x^{\mathrm{T}}(t)\Lambda x(t)\mathrm{d}t\nonumber \\
 & \,\,\quad+(x^{\mathrm{T}}(t)RBu(t)+u^{\mathrm{T}}(t)B^{\mathrm{T}}Rx(t))\mathrm{d}t\nonumber \\
 & \,\,\quad+\sum_{i=1}^{\ell_{1}}x^{\mathrm{T}}(t)(C_{i}^{\mathrm{T}}R+RC_{i})x(t)\mathrm{d}W_{i}(t)\nonumber \\
 & \,\,\quad+\sum_{i=1}^{\ell_{2}}(x^{\mathrm{T}}(t)Rd_{i}+d_{i}^{\mathrm{T}}Rx(t))\mathrm{d}W_{i+\ell_{1}}(t),\label{eq:ito-lemma}
\end{align}
 where $d_{i}\in\mathbb{R}^{n}$, $i\in\{1,\ldots,\ell_{2}\}$, denote
the columns of matrix $D$. Under an $\mathcal{F}_{t}$-adapted control
policy, $\{x(t)\}_{t\geq0}$ is $\mathcal{F}_{t}$-adapted. Thus,
by Theorem 3.2.1 of \cite{oksendal2005stochastic}, we have $\mathbb{E}[\int_{0}^{t}x^{\mathrm{T}}(\tau)(C_{i}^{\mathrm{T}}R+RC_{i})x(\tau)\mathrm{d}W_{i}(\tau)]=0$
and $\mathbb{E}[\int_{0}^{t}(x^{\mathrm{T}}(\tau)Rd_{i}+d_{i}^{\mathrm{T}}Rx(\tau))\mathrm{d}W_{i+\ell_{1}}(\tau)]=0$.
As a result, it follows from \eqref{eq:ito-lemma} that 
\begin{align}
 & \mathbb{E}[V(x(t))]\nonumber \\
 & \,=\mathbb{E}[V(x(0))]+\mathrm{tr}(D^{\mathrm{T}}RD)t+\mathbb{E}\left[\int_{0}^{t}x^{\mathrm{T}}(\tau)\Lambda x(\tau)\mathrm{d}\tau\right]\nonumber \\
 & \,\quad+\mathbb{E}\left[\int_{0}^{t}(x^{\mathrm{T}}(\tau)RBu(\tau)+u^{\mathrm{T}}(\tau)B^{\mathrm{T}}Rx(\tau))\mathrm{d}\tau\right],\label{eq:pre-ev}
\end{align}
 for $t\geq0$. Next, we change the order of expectation and integration
in \eqref{eq:pre-ev} by using Fubini's theorem \cite{billingsley1986}
and obtain
\begin{align}
 & \mathbb{E}[V(x(t))]\nonumber \\
 & \,=\mathbb{E}[V(x(0))]+\mathrm{tr}(D^{\mathrm{T}}RD)t+\int_{0}^{t}\mathbb{E}[x^{\mathrm{T}}(\tau)\Lambda x(\tau)]\mathrm{d}\tau\nonumber \\
 & \,\quad+\int_{0}^{t}\mathbb{E}[x^{\mathrm{T}}(\tau)RBu(\tau)+u^{\mathrm{T}}(\tau)B^{\mathrm{T}}Rx(\tau)]\mathrm{d}\tau.\label{eq:ev}
\end{align}
In what follows, we use \eqref{eq:ev} to show $\lim_{t\to\infty}\mathbb{E}\left[V(x(t))\right]=\infty$,
separately for two cases: $\beta_{\mathrm{U}}=0$ and $\beta_{\mathrm{U}}>0$.

First, consider the case where $\beta_{\mathrm{U}}=\lambda_{\max}(B^{\mathrm{T}}RB)=0$.
In this case, we have $R^{\frac{1}{2}}B=0$, and hence $RB=0$. Furthermore,
\eqref{eq:A-cond-1} implies $\mathbb{E}[x^{\mathrm{T}}(\tau)\Lambda x(\tau)]\geq\phi_{\mathrm{L}}\mathbb{E}[x^{\mathrm{T}}(\tau)Rx(\tau)]=\phi_{\mathrm{L}}\mathbb{E}[V(x(\tau))]$.
As a consequence, we obtain from \eqref{eq:ev} that $\mathbb{E}[V(x(t))]\geq\mathbb{E}[V(x(0))]+\mathrm{tr}(D^{\mathrm{T}}RD)t+\phi_{\mathrm{L}}\int_{0}^{t}\mathbb{E}[V(x(\tau))]\mathrm{d}\tau$
for all $t\geq0$. Therefore, we can use Lemma~\ref{LemmaInverseGronwall}
with $c_{0}=\mathbb{E}[V(x(0))]$, $c_{1}=\mathrm{tr}(D^{\mathrm{T}}RD)$,
$\phi=\phi_{\mathrm{L}}$, and $h(t)=\mathbb{E}[V(x(t))]$ to obtain
\begin{align}
\mathbb{E}[V(x(t))] & \geq\mathbb{E}[V(x(0))]e^{\phi_{\mathrm{L}}t}\nonumber \\
 & \quad+(\mathrm{tr}(D^{\mathrm{T}}RD)/\phi_{\mathrm{L}})(e^{\phi_{\mathrm{L}}t}-1).\label{eq:ev-ineq}
\end{align}
Notice that $\phi_{\mathrm{L}}$ is positive, and hence, $\lim_{t\to\infty}e^{\phi_{\mathrm{L}}t}=\infty$.
Moreover, $\mathrm{tr}(D^{\mathrm{T}}RD)$ is positive by the assumption
\eqref{eq:nonzero-noise-1}. As a result, \eqref{eq:ev-ineq} implies
$\lim_{t\to\infty}\mathbb{E}\left[V(x(t))\right]=\infty$. 

Next, we will show that $\lim_{t\to\infty}\mathbb{E}\left[V(x(t))\right]=\infty$
holds also for the case where $\beta_{\mathrm{U}}>0$. For this case
let $\kappa\triangleq\phi_{\mathrm{L}}\mathrm{tr}\left(D^{\mathrm{T}}RD\right)/\beta_{\mathrm{U}}$
and $\mathcal{Q}\triangleq\left\{ q>1\colon\hat{u}<\kappa/q\right\} $.
By Lemma~\ref{Lemma-HatU}, we have $\mathcal{Q}\neq\emptyset$. 

Now let $\hat{q}\in\mathcal{Q}$ and define $\gamma(\hat{q})\triangleq\hat{q}/\phi_{\mathrm{L}}$.
The scalars $\gamma^{1/2}(\hat{q})$ and $\gamma^{-1/2}(\hat{q})$
are well-defined since $\gamma(\hat{q})>0$. Moreover, since $R$
is a nonnegative-definite Hermitian matrix, we have $0\leq z^{\mathrm{T}}Rz$
for any $z\in\mathbb{R}^{n}$. Using this inequality with $z=\gamma^{-1/2}(q)x(\tau)+\gamma^{1/2}(q)Bu(\tau)$,
we get 
\begin{align*}
0 & \leq\left(\gamma^{-1/2}(q)x(\tau)+\gamma^{1/2}(q)Bu(\tau)\right)^{\mathrm{T}}R\\
 & \quad\cdot\left(\gamma^{-1/2}(q)x(\tau)+\gamma^{1/2}(q)Bu(\tau)\right)\\
 & =\gamma^{-1}(\hat{q})x^{\mathrm{T}}(\tau)Rx(\tau)+x^{\mathrm{T}}(\tau)RBu(\tau)\\
 & \quad+u^{\mathrm{T}}(\tau)B^{\mathrm{T}}Rx(\tau)+\gamma(\hat{q})u_{1}^{\mathrm{T}}(\tau)B^{\mathrm{T}}RBu(\tau),
\end{align*}
which implies
\begin{align}
 & x^{\mathrm{T}}(\tau)RBu(\tau)+u^{\mathrm{T}}(\tau)B^{\mathrm{T}}Rx(\tau)\nonumber \\
 & \,\geq-\gamma^{-1}(\hat{q})x^{\mathrm{T}}(\tau)Rx(\tau)-\gamma(\hat{q})u^{\mathrm{T}}(\tau)B^{\mathrm{T}}RBu(\tau).\label{eq:gamma-ineq}
\end{align}
  It then follows from \eqref{eq:ev} together with $\mathbb{E}[x^{\mathrm{T}}(\tau)\Lambda x(\tau)]\geq\phi_{\mathrm{L}}\mathbb{E}[V(x(\tau))]$
and \eqref{eq:gamma-ineq} that 
\begin{align}
\mathbb{E}[V(x(t))] & \geq\mathbb{E}[V(x(0))]+\mathrm{tr}(D^{\mathrm{T}}RD)t\nonumber \\
 & \quad+(\phi_{\mathrm{L}}-\gamma^{-1}(\hat{q}))\int_{0}^{t}\mathbb{E}[V(x(\tau))]\mathrm{d}\tau\nonumber \\
 & \quad-\gamma(\hat{q})\int_{0}^{t}\mathbb{E}[u^{\mathrm{T}}(\tau)B^{\mathrm{T}}RBu(\tau)]\mathrm{d}\tau.\label{eq:ev-phi}
\end{align}
Since $\gamma(\hat{q})>0$, we have $-\gamma(\hat{q})<0$. Thus, by
using $u^{\mathrm{T}}(\tau)B^{\mathrm{T}}RBu(\tau)\leq\lambda_{\max}(B^{\mathrm{T}}RB)\|u(\tau)\|^{2}=\beta_{\mathrm{U}}\|u(\tau)\|^{2}$
with \eqref{eq:ev-phi}, we obtain
\begin{align}
\mathbb{E}[V(x(t))] & \geq\mathbb{E}[V(x(0))]+\mathrm{tr}(D^{\mathrm{T}}RD)t\nonumber \\
 & \quad+(\phi_{\mathrm{L}}-\gamma^{-1}(\hat{q}))\int_{0}^{t}\mathbb{E}[V(x(\tau))]\mathrm{d}\tau\nonumber \\
 & \quad-\gamma(\hat{q})\beta_{\mathrm{U}}\int_{0}^{t}\mathbb{E}[\|u(\tau)\|^{2}]\mathrm{d}\tau.\label{eq:ev-phi-next}
\end{align}
 Now, since the control policy satisfies \eqref{eq:average-u-bound-1},
we have $\int_{0}^{t}\mathbb{E}[\|u(\tau)\|^{2}]\mathrm{d}\tau\leq\hat{u}t$,
and hence, it follows from \eqref{eq:ev-phi-next} that 
\begin{align}
\mathbb{E}[V(x(t))] & \geq\mathbb{E}[V(x(0))]+\mathrm{tr}(D^{\mathrm{T}}RD)t\nonumber \\
 & \quad+(\phi_{\mathrm{L}}-\gamma^{-1}(\hat{q}))\int_{0}^{t}\mathbb{E}[V(x(\tau))]\mathrm{d}\tau\nonumber \\
 & \quad-\gamma(\hat{q})\beta_{\mathrm{U}}\hat{u}t,\quad t\geq0.\label{eq:ev-phi-final}
\end{align}
 Let $c_{0}\triangleq\mathbb{E}[V(x(0))]$, $c_{1}\triangleq\mathrm{tr}(D^{\mathrm{T}}RD)-\gamma(\hat{q})\beta_{\mathrm{U}}\hat{u}$,
$\phi\triangleq\phi_{\mathrm{L}}-\gamma^{-1}(\hat{q})$, and $h(t)\triangleq\mathbb{E}[V(x(t))]$.
By definition of $\mathcal{Q}$, we have $\hat{q}>1$, which implies
$1-1/\hat{q}>0$. This inequality and $\phi_{\mathrm{L}}>0$ imply
$\phi=\phi_{\mathrm{L}}-\gamma^{-1}(\hat{q})=\phi_{\mathrm{L}}(1-1/\hat{q})>0$.
Thus, by Lemma~\ref{LemmaInverseGronwall}, we obtain
\begin{align}
\mathbb{E}[V(x(t))] & \geq c_{0}e^{\phi t}+(c_{1}/\phi)(e^{\phi t}-1).\label{eq:V-ineq}
\end{align}
 Since $V$ is a nonnegative-definite function, we have $c_{0}\geq0$.
Next, we show $c_{1}>0$. By definition of $\mathcal{Q}$, we have
$\hat{u}<\phi_{\mathrm{L}}\mathrm{tr}\left(D^{\mathrm{T}}RD\right)/(\beta_{\mathrm{U}}\hat{q})$.
Noting that $\gamma(\hat{q})=\hat{q}/\phi_{\mathrm{L}}$, this inequality
implies $\gamma(\hat{q})\beta_{\mathrm{U}}\hat{u}<\mathrm{tr}\left(D^{\mathrm{T}}RD\right)$.
Thus, $c_{1}=\mathrm{tr}(D^{\mathrm{T}}RD)-\gamma(\hat{q})\beta_{\mathrm{U}}\hat{u}>0$.
Now, since $c_{0}\geq0$, $c_{1}>0$, and $\phi>0$ hold, \eqref{eq:V-ineq}
implies $\lim_{t\to\infty}\mathbb{E}\left[V(x(t))\right]=\infty$. 

Finally, since $R\in\mathbb{C}^{n\times n}\setminus\{0\}$ (i.e.,
$R\neq0$), the nonnegative-definite Hermitian matrix $R$ has at
least one eigenvalue strictly larger than $0$. Thus, $\lambda_{\max}(R)>0$.
 Consequently, $V(x(t))\leq\lambda_{\max}(R)\|x(t)\|^{2}$ implies
$\|x(t)\|^{2}\geq\left(1/\lambda_{\max}(R)\right)V(x(t))$ for $t\geq0$.
Hence, \eqref{eq:divergence-1} follows from $\lim_{t\to\infty}\mathbb{E}\left[V(x(t))\right]=\infty$.
\end{proof} 

Theorem~\ref{TheoremContinuousTime} provides sufficient conditions
under which the system \eqref{eq:continuous-time-system} is instabilizable
and the second moment of the state diverges regardless of the controller
design and the initial state value. Condition \eqref{eq:A-cond-1}
in Theorem~\ref{TheoremContinuousTime} quantifies the instability
of the uncontrolled ($u(t)\equiv0$) system, and the term $\mathrm{tr}\left(D^{\mathrm{T}}RD\right)$
in \eqref{eq:nonzero-noise-1} represents the effect of additive noise
characterized with the matrix $D$. If there is no multiplicative
noise (i.e., $C_{i}=0$ for $i\in\{1,\ldots,\ell_{1}\}$), then \eqref{eq:A-cond-1}
requires $A$ to be strictly unstable. On the other hand, when there
is multiplicative noise, \eqref{eq:A-cond-1} may hold even if $A$
is Hurwitz-stable. Notice also that $R$ is a nonnegative-definite
matrix and it may have $0$ as an eigenvalue. This property is essential
in our analysis, since it allows us to deal with the cases where some
of the states are diverging, while the others are stable. 

Theorem~\ref{TheoremContinuousTime} implies that if conditions \eqref{eq:A-cond-1},
\eqref{eq:nonzero-noise-1} hold, then it is \emph{not possible} to
stabilize the system by using control inputs with too small average
second moments as in \eqref{eq:average-u-bound-1}. The impossibility
threshold on the average second moment of control input $u(t)$ is
characterized in \eqref{eq:ubar-cond-1}. If $\lambda_{\max}(B^{\mathrm{T}}RB)=0$,
then this threshold value becomes infinity indicating that stabilization
is impossible regardless of the input constraint. We note that the
case $\lambda_{\max}(B^{\mathrm{T}}RB)=0$ represents the situation,
where $u(t)$ does not have any effect on $x^{\mathrm{T}}(t)Rx(t)$. 

\begin{remark}[Instability conditions]The structure of condition
\eqref{eq:A-cond-1} is similar to those of stability/instability
conditions provided in \cite{mao2007stochastic} for stochastic systems
with multiplicative noise. In particular, when specialized to linear
systems, Corollary~4.7 of \cite{mao2007stochastic} yields an instability
condition based on existence of a positive-definite matrix $P\in\mathbb{R}^{n\times n}$
and a scalar $\psi>0$ such that $A^{\mathrm{T}}P+PA+\sum_{i=1}^{\ell_{1}}C_{i}^{\mathrm{T}}PC_{i}\geq\psi P$.
Notice that for systems with only multiplicative noise, a positive-definite
matrix $P$ is required to show global instability. In our setting,
a nonnegative-definite matrix $R$ is sufficient, because there is
also additive noise and \eqref{eq:nonzero-noise-1} guarantees that
this noise can make the projection of the state on unstable modes
of the uncontrolled system take a nonzero value even if the initial
state is zero. Moreover, under the condition \eqref{eq:ubar-cond-1},
$\mathbb{E}[x^{\mathrm{T}}(t)Rx(t)]$ diverges, which in turn implies
divergence of the second moment of the state, as shown in the proof
of Theorem~\ref{TheoremContinuousTime}. \hfill $\triangleleft$
\end{remark} 

\begin{remark}[Numerical approach] \label{RemarkNumerical} We note
that linear matrix inequalities can be used for checking the conditions
of Theorem~\ref{TheoremContinuousTime}. First of all, for a given
$\phi_{\mathrm{L}}$, condition \eqref{eq:A-cond-1} is linear in
$R$. Similarly, \eqref{eq:nonzero-noise-1} is a linear inequality
of $R$. Note that \eqref{eq:nonzero-noise-1} also guarantees that
$R\neq0$. Moreover, the inequality \eqref{eq:ubar-cond-1} can be
transformed into $\overline{\beta}\hat{u}<\phi_{\mathrm{L}}\mathrm{tr}\left(D^{\mathrm{T}}RD\right)$
and $B^{\mathrm{T}}RB\leq\overline{\beta}I_{m},$ which are linear
in $R$ and $\overline{\beta}\geq0$, for a given $\phi_{\mathrm{L}}$.
If the abovementioned inequalities are satisfied with $R=\widetilde{R}$,
then $R=c\widetilde{R}$ with any $c>0$ also satisfies them. To restrict
the solutions, we can impose an additional constraint $\mathrm{tr}(R)=1$.
In our numerical method, we iterate over a set of candidate values
of $\phi_{\mathrm{L}}$ and utilize linear matrix inequality solvers
(for each value of $\phi_{\mathrm{L}}$) to check the conditions of
Theorem~\ref{TheoremContinuousTime}. \hfill $\triangleleft$\end{remark}

\begin{remark}[Partial constraints] Theorem~\ref{TheoremContinuousTime}
can be extended to handle partial input constraints. Consider  $\mathrm{d}x(t)=(Ax(t)+Bu(t)+F\nu(t))\mathrm{d}t+\left[\Psi(x(t)),D\right]\mathrm{d}W(t)$,
where $u(t)$ is constrained as in \eqref{eq:average-u-bound-1} and
$\nu(t)$ is \emph{unconstrained}. If \eqref{eq:A-cond-1}--\eqref{eq:ubar-cond-1}
and $RF=0$ hold, then it is impossible to achieve stabilization of
this modified system. The proof is similar to that of Theorem~\ref{TheoremContinuousTime},
as $RF=0$ implies that $V(x(t))=x^{\mathrm{T}}(t)Rx(t)$ is not affected
by $\nu(\cdot)$, and hence \eqref{eq:ito-lemma} holds. \hfill $\triangleleft$\end{remark} 

\subsection{Tightness of the result for scalar systems}

Theorem~\ref{TheoremContinuousTime} provides a tight bound for $\hat{u}$
in \eqref{eq:ubar-cond-1} for scalar systems with a scalar state
and a scalar constrained input. 

Consider \eqref{eq:continuous-time-system} with scalars $A,B,D$
and scalar-valued function $\Psi(x)\triangleq C_{1}x$ such that $2A+C_{1}^{2}>0$,
$B\neq0$, and $D\neq0$. In this case, conditions \eqref{eq:A-cond-1}
and \eqref{eq:nonzero-noise-1} hold with $R=1$ and $\phi_{\mathrm{L}}=2A+C_{1}^{2}$.
Thus, Theorem~\ref{TheoremContinuousTime} implies that if the control
policy satisfies \eqref{eq:average-u-bound-1} with $\hat{u}<(2A+C_{1}^{2})D^{2}/B^{2}$,
then the system is impossible to stabilize regardless of the initial
state $x_{0}$. This bound is tight, because, as shown in the following
result, stability can be achieved when $\hat{u}=(2A+C_{1}^{2})D^{2}/B^{2}$. 

\begin{proposition} \label{PropositionScalar} Consider \eqref{eq:continuous-time-system}
with scalars $A,B,D\in\mathbb{R}$ and scalar-valued function $\Psi(x)\triangleq C_{1}x$.
Suppose $2A+C_{1}^{2}>0$, $B\neq0$, $D\neq0$, and $x_{0}=0$. Then
feedback control policy $u(t)=Kx(t)$ with $K=-(2A+C_{1}^{2})/B$
can achieve stabilization (i.e., $\sup_{t\geq0}\mathbb{E}[x^{2}(t)]<\infty$)
and satisfies \eqref{eq:average-u-bound-1} with $\hat{u}=(2A+C_{1}^{2})D^{2}/B^{2}$.
\end{proposition}

\begin{proof}Let $\Lambda\triangleq2(A+BK)+C_{1}^{2}$ and $\phi(x,t)\triangleq e^{-\Lambda t}x^{2}$.
By Ito formula (Section~4.2 of \cite{oksendal2005stochastic}), 
\begin{align}
\mathrm{d}\phi(x(t),t) & =2e^{-\Lambda t}C_{1}x^{2}(t)\mathrm{d}W_{1}(t)+2e^{-\Lambda t}Dx(t)\mathrm{d}W_{2}(t)\nonumber \\
 & \quad+e^{-\Lambda t}D^{2}\mathrm{d}t.\label{eq:dphi}
\end{align}
Now, since $\{x(t)\}_{t\geq0}$ is an $\mathcal{F}_{t}$-adapted process,
we obtain $\mathbb{E}[\int_{0}^{t}2e^{-\Lambda\tau}C_{1}x^{2}(\tau)\mathrm{d}W_{1}(\tau)]=0$
and $\mathbb{E}[\int_{0}^{t}2e^{-\Lambda\tau}Dx(\tau)\mathrm{d}W_{2}(\tau)]=0$,
by using Theorem~3.2.1 of \cite{oksendal2005stochastic}. Thus, with
$x_{0}=0$, \eqref{eq:dphi} implies $\mathbb{E}[\phi(x(t),t)]=\mathbb{E}[\phi(x(0),0)]+\int_{0}^{t}e^{-\Lambda\tau}D^{2}\mathrm{d}\tau=-(D^{2}/\Lambda)(e^{-\Lambda t}-1)$.
Therefore, for all $t\geq0$, we have $\mathbb{E}[x^{2}(t)]=e^{\Lambda t}\mathbb{E}[\phi(x(t),t)]=-(D^{2}/\Lambda)(1-e^{\Lambda t})\leq-D^{2}/\Lambda$,
where the last inequality follows from $\Lambda=-(2A+C_{1}^{2})<0$.
As a consequence,  $\sup_{t\geq0}\mathbb{E}[x^{2}(t)]\leq-D^{2}/\Lambda<\infty$,
which implies that stability is achieved. Moreover, we have $\mathbb{E}[u^{2}(t)]=K^{2}\mathbb{E}[x^{2}(t)]\leq(2A+C_{1}^{2})D^{2}/B^{2}$
for all $t\geq0$, showing that control input constraint \eqref{eq:average-u-bound-1}
is satisfied with $\hat{u}=(2A+C_{1}^{2})D^{2}/B^{2}$. \end{proof}

Proposition~\ref{PropositionScalar} handles the case where $2A+C_{1}^{2}>0$.
With similar analysis, we can also show that if $2A+C_{1}^{2}<0$,
then $\mathbb{E}[x^{2}(t)]$ stays bounded even without control (i.e.,
$u(t)\equiv0$). Furthermore, if $2A+C_{1}^{2}=0$, $B\neq0$, then
a state-feedback control policy $u(t)=Kx(t)$ with $BK<0$ achieves
stabilization, and moreover, for $x_{0}=0$, it guarantees the bound
$\mathbb{E}[u^{2}(t)]\leq-(1/2)D^{2}K/B$ for all $t\geq0$. This
bound can be made arbitrarily small by choosing small $|K|$. If $2A+C_{1}^{2}=0$,
$B=0$, then the system is uncontrollable and $\mathbb{E}[x^{2}(t)]$
grows unboundedly unless $D=0$. 

\subsection{Existence of instability-inducing noise matrices }

 The following proposition complements Theorem~\ref{TheoremContinuousTime}.
It shows that if $A,C_{1},\ldots,C_{\ell_{1}}$ satisfy \eqref{eq:A-cond-1},
then for any $B$ and $\hat{u}$, there exists a noise matrix $D$
that satisfies both \eqref{eq:nonzero-noise-1} and \eqref{eq:ubar-cond-1}.
Thus, by Theorem~\ref{TheoremContinuousTime}, the system with that
noise matrix is impossible to be stabilized under constraint \eqref{eq:average-u-bound-1}. 

\begin{proposition} \label{Proposition-D-destabilization} Assume
that there exist a nonnegative-definite Hermitian matrix $R\in\mathbb{C}^{n\times n}\setminus\{0\}$
and a scalar $\phi_{\mathrm{L}}>0$ that satisfy \eqref{eq:A-cond-1}.
Then for any $B\in\mathbb{R}^{n\times m}$ and $\hat{u}\in[0,\infty)$,
there exists $D\in\mathbb{R}^{n\times\ell_{2}}$ such that both \eqref{eq:nonzero-noise-1}
and \eqref{eq:ubar-cond-1} hold. \end{proposition} 


\begin{inarxiv}\begin{proof} By the spectral theorem for Hermitian
matrices (see Theorem~2.5.6 of \cite{hornmatrixanalysis}), the nonnegative-definite
Hermitian matrix $R$ can be written as $R=\Xi\mathrm{diag}(\mu_{1},\ldots,\mu_{n})\Xi^{*}$,
where $\mu_{1},\ldots,\mu_{n}\geq0$ are the eigenvalues of $R$ and
$\Xi\in\mathbb{C}^{n\times n}\setminus\{0\}$ is a unitary matrix.
Let $\xi_{1},\ldots,\xi_{n}\in\mathbb{C}^{n}\setminus\{0\}$ denote
the columns of $\Xi$. We have $R=\sum_{i}\mu_{i}\xi_{i}\xi_{i}^{*}$.
Since $R\neq0$, at least one eigenvalue of $R$ is strictly larger
than $0$. Let $\tilde{i}\triangleq\min\{i\in\{1,\ldots,n\}\colon\mu_{i}>0\}$.
Now, let $\xi_{\tilde{i},j}\in\mathbb{C}$ denote the $j$th entry
of vector $\xi_{\tilde{i}}$. Since $\xi_{\tilde{i}}\neq0$, at least
one entry of $\xi_{\tilde{i}}$ is nonzero. We define $\tilde{j}\triangleq\min\{j\in\{1,\ldots,n\}\colon\xi_{\tilde{i},j}\neq0\}$.
Now let $\alpha>0$ and $D\triangleq[d_{1},\ldots,d_{\ell_{2}}]$,
where $d_{i}\in\mathbb{R}^{n}$, $i\in\{1,\ldots,\ell_{2}\}$, are
columns of $D$. We let $d_{i}=0$ for $i\neq1$, and set the entries
of $d_{1}\in\mathbb{R}^{n}$ as 
\begin{align}
d_{1,j} & =\begin{cases}
\frac{(\alpha+(\hat{u}\beta_{\mathrm{U}}/\phi_{\mathrm{L}}))^{1/2}}{\mu_{\tilde{i}}^{1/2}|\xi_{\tilde{i},\tilde{j}}|}, & \mathrm{if}\,\,j=\tilde{j},\\
0, & \mathrm{otherwise},
\end{cases}\label{eq:d1-def}
\end{align}
where $\beta_{\mathrm{U}}=\lambda_{\max}(B^{\mathrm{T}}RB)$. Since
$d_{i}=0$ for $i\neq1$, we have 
\begin{align}
 & \mathrm{tr}(D^{\mathrm{T}}RD)=\mathrm{tr}\left(\sum_{i=1}^{n}\mu_{i}D^{\mathrm{T}}\xi_{i}\xi_{i}^{*}D\right)\nonumber \\
 & \quad=\sum_{i=1}^{n}\mu_{i}\mathrm{tr}(D^{\mathrm{T}}\xi_{i}\xi_{i}^{*}D)=\sum_{i=1}^{n}\mu_{i}\sum_{j=1}^{\ell_{2}}\left(d_{j}^{\mathrm{T}}\xi_{i}\xi_{i}^{*}d_{\mathrm{j}}\right)\nonumber \\
 & \quad=\sum_{i=1}^{n}\mu_{i}d_{1}^{\mathrm{T}}\xi_{i}\xi_{i}^{*}d_{1}=d_{1}^{\mathrm{T}}\left(\sum_{i=1}^{n}\mu_{i}\xi_{i}\xi_{i}^{*}\right)d_{1}.\label{eq:d1-ineq}
\end{align}
Moreover, since $\sum_{i=1}^{n}\mu_{i}\xi_{i}\xi_{i}^{*}\geq\mu_{\tilde{i}}\xi_{\tilde{i}}\xi_{\tilde{i}}^{*}$,
it follows from \eqref{eq:d1-ineq} that $\mathrm{tr}(D^{\mathrm{T}}RD)\geq\mu_{\tilde{i}}d_{1}^{\mathrm{T}}\xi_{\tilde{i}}\xi_{\tilde{i}}^{*}d_{1}$.
By using this inequality and noting that $\mu_{\tilde{i}}>0$, we
obtain from \eqref{eq:d1-def} that 
\begin{align}
 & \mathrm{tr}(D^{\mathrm{T}}RD)\geq\mu_{\tilde{i}}\left(\sum_{j=1}^{n}d_{1,j}\xi_{\tilde{i},\tilde{j}}\right)\left(\sum_{j=1}^{n}\overline{\xi}_{\tilde{i},\tilde{j}}d_{1,j}\right)\nonumber \\
 & \quad=\mu_{\tilde{i}}\xi_{\tilde{i},\tilde{j}}\overline{\xi}_{\tilde{i},\tilde{j}}d_{1,\tilde{j}}^{2}=\mu_{\tilde{i}}|\xi_{\tilde{i},\tilde{j}}|^{2}\frac{\alpha+(\hat{u}\beta_{\mathrm{U}}/\phi_{\mathrm{L}})}{\mu_{\tilde{i}}|\xi_{\tilde{i},\tilde{j}}|^{2}}\nonumber \\
 & \quad=\alpha+(\hat{u}\beta_{\mathrm{U}}/\phi_{\mathrm{L}}),\label{eq:tr-ineq}
\end{align}
which implies \eqref{eq:nonzero-noise-1}, as $\alpha>0$, $\hat{u}\geq0$,
$\beta_{\mathrm{U}}\geq0$, and $\phi_{\mathrm{L}}>0$. If $\beta_{\mathrm{U}}=0$,
then the inequality \eqref{eq:ubar-cond-1} holds directly, as $\hat{u}<\infty$.
If $\beta_{\mathrm{U}}\neq0$, then the inequality \eqref{eq:tr-ineq}
implies $\mathrm{tr}(D^{\mathrm{T}}RD)>\hat{u}\beta_{\mathrm{U}}/\phi_{\mathrm{L}}$,
which in turn implies \eqref{eq:ubar-cond-1}. \end{proof} \end{inarxiv}

\section{Special Setting with Only Additive Noise }

\label{sec:Special-Setting-with}

In this section, we are interested in a special setting, where $\Psi(x(t))=0$
in \eqref{eq:continuous-time-system}. In this setting, the system
does not face multiplicative noise and it is only subject to additive
noise. 

We first present an improvement of the numerical approach presented
in Remark~\ref{RemarkNumerical}. Then we show that eigenstructure
of $A$ can be used to obtain new instability conditions that are
easier to check compared to Theorem~\ref{TheoremContinuousTime}.
The eigenstructure-based analysis was previously considered only for
discrete-time systems in \cite{cetinkaya2020TAC}. Here, we show that
continuous-time systems also allow a similar approach. 

\subsection{Candidate values of $\phi_{\mathrm{L}}$ in instability analysis}

Remark~\ref{RemarkNumerical} provides a numerical approach for checking
instability conditions of Theorem~\ref{TheoremContinuousTime}. This
approach is based on iterating over a set of candidate values of $\phi_{\mathrm{L}}$
and checking the feasibility of certain linear matrix inequalities.
In the general case with both additive and multiplicative noise, we
do not have prespecified bounds for the candidate value set. However,
in the case of only additive noise ($\Psi(x(t))=0$, i.e., $C_{i}=0$
in \eqref{eq:Psi-C-def}), the candidate values of $\phi_{\mathrm{L}}$
can be restricted to belong to the set $(0,2\vartheta_{\max}(A)]$,
where we define $\vartheta_{\max}\colon\mathbb{C}^{n\times n}\to\mathbb{R}$
as 
\begin{align}
\vartheta_{\max}(A) & \triangleq\max\{\mathrm{Re}(\lambda)\colon\lambda\in\mathrm{spec}(A)\}\label{eq:vartheta-max-def}
\end{align}
 with $\mathrm{spec}(A)\subset\mathbb{C}$ denoting the set of eigenvalues
of $A$. It is sufficient to choose the values of $\phi_{\mathrm{L}}$
from the set $(0,2\vartheta_{\max}(A)]$, as it is not possible to
satisfy \eqref{eq:A-cond-1} with $\phi_{\mathrm{L}}>2\vartheta_{\max}(A)$
and $R\neq0$, as shown in the following proposition. Here, we note
that with $C_{i}=0$, $i\in\{1,\ldots,\ell_{1}\}$, the inequality
\eqref{eq:A-cond-1} reduces to $A^{\mathrm{T}}R+RA\geq\phi_{\mathrm{L}}R$. 

\begin{proposition} Let $A\in\mathbb{C}^{n\times n}$. For every
nonnegative-definite Hermitian matrix $R\in\mathbb{C}^{n\times n}\setminus\{0\}$,
there exists $y\in\mathbb{C}^{n}\setminus\{0\}$ such that $R^{\frac{1}{2}}y\neq0$
and $y^{*}(A^{\mathrm{*}}R+RA)y\leq2\vartheta_{\max}(A)y^{*}Ry$,
where $\vartheta_{\max}(A)$ is defined in \eqref{eq:vartheta-max-def}.
\end{proposition}

\begin{proof}It follows from Lemma~A.1 of \cite{cetinkaya2020TAC}
that $R^{\frac{1}{2}}A\hat{\nu}=\hat{\lambda}R^{\frac{1}{2}}\hat{\nu}$,
where $\hat{\lambda}\in\mathbb{C}$ is an eigenvalue of $A$ and $\hat{\nu}\in\mathbb{C}^{n}\setminus\{0\}$
is a generalized eigenvector of $A$ that satisfies $R^{\frac{1}{2}}\hat{\nu}\neq0$.
Let $y\triangleq\hat{\nu}$. We have $R^{\frac{1}{2}}y\neq0$. Moreover,
$y^{*}(A^{\mathrm{*}}R+RA)y=\overline{\hat{\lambda}}y^{*}Ry+\hat{\lambda}y^{*}Ry=2\mathrm{Re}(\hat{\lambda})y^{*}Ry$.
Therefore, $y^{*}(A^{\mathrm{*}}R+RA)y\leq2\vartheta_{\max}(A)y^{*}Ry,$
since $\mathrm{Re}(\hat{\lambda})\leq\vartheta_{\max}(A)$. \end{proof}

\subsection{Instability conditions based on the eigenstructure of $A$}

Even with the improvement discussed in the previous subsection, checking
feasibility of the linear matrix inequalities mentioned in Remark~\ref{RemarkNumerical}
can be computationally costly. In this subsection, we show that the
eigenstructure of the system matrix $A$ can be used to derive instability
conditions that can be checked numerically efficiently. 

Let $r\in\{1,\ldots,n\}$ denote the number of distinct eigenvalues
of $A$ and let $\lambda_{1},\lambda_{2},\ldots,\lambda_{r}\in\mathbb{C}$
with $\lambda_{i}\neq\lambda_{j}$ denote those eigenvalues. Moreover,
for every $i\in\{1,\ldots,r\}$, let $n_{i}\in\{1,\ldots,n\}$ represent
the geometric multiplicity of the eigenvalue $\lambda_{i}$. Eigenvalues
of $A$ are also the eigenvalues of $A^{\mathrm{T}}$ with the same
multiplicities. Thus, for every $\lambda_{i}$, there exists $n_{i}$
number of vectors $v_{i,j}\in\mathbb{C}^{n}$ such that
\begin{align}
A^{\mathrm{T}}v_{i,j} & =\lambda_{i}v_{i,j},\quad j\in\{1,\ldots,n_{i}\}.\label{eq:eigenvalue-eigenvector-pairs}
\end{align}
 These vectors $v_{i,1},\ldots,v_{i,n_{i}}$ are called the left-eigenvectors
of $A$ associated with the eigenvalue $\lambda_{i}$. We remark that
if $\lambda_{i}$ is a complex eigenvalue (i.e., $\lambda_{i}\notin\mathbb{R}$),
then the complex conjugate $\overline{\lambda_{i}}$ is also an eigenvalue
of $A$, and moreover, 
\begin{align}
v_{i,j}^{*}A & =\overline{\lambda_{i}}v_{i,j}^{*},\quad j\in\{1,\ldots,n_{i}\},\label{eq:left-eigenvector}
\end{align}
where $v_{i,j}^{\mathrm{*}}$ is the complex conjugate transpose of
$v_{i,j}\in\mathbb{C}^{n}$. In the instability conditions presented
below, we use the eigenvalues $\lambda_{i}$, and the left-eigenvectors
$v_{i,j}\in\mathbb{C}^{n}$, $j\in\{1,\ldots,n_{i}\}$, $i\in\{1,\ldots,r\}$.
Furthermore, we define 
\begin{align*}
 & \mathcal{I}\triangleq\{(i,j)\colon\mathrm{\mathrm{Re}(\lambda_{i})}>0,\,\,v_{i,j}^{*}DD^{\mathrm{T}}v_{i,j}>0,\\
 & \quad\,\,\,\,\,\,\,\,j\in\{1,\ldots,n_{i}\},\,i\in\{1,\ldots,r\}\}.
\end{align*}

\begin{corollary} \label{Corollary-Eigen} Consider the linear stochastic
system \eqref{eq:continuous-time-system} where $\Psi(x)=0$. Suppose
$\mathcal{I}\neq\emptyset$. If the control policy is $\mathcal{F}_{t}$-adapted
and satisfies \eqref{eq:average-u-bound-1} with 
\begin{align}
\hat{u} & <\max_{(i,j)\in\mathcal{I}}\varphi_{i,j},\label{eq:ubarcond2}
\end{align}
 where 
\begin{align*}
\varphi_{i,j} & \triangleq\begin{cases}
2\mathrm{Re}(\lambda_{i})v_{i,j}^{*}DD^{\mathrm{T}}v_{i,j}/\beta_{\mathrm{U}}^{i,j}, & \mathrm{if}\,\,\beta_{\mathrm{U}}^{i,j}\neq0,\\
\infty, & \mathrm{otherwise},
\end{cases}
\end{align*}
and $\beta_{\mathrm{U}}^{i,j}\triangleq\lambda_{\max}(B^{\mathrm{T}}v_{i,j}v_{i,j}^{*}B)$,
$(i,j)\in\mathcal{I}$, then the second moment of the state diverges
(i.e., \eqref{eq:divergence-1} holds) for any initial state $x_{0}\in\mathbb{R}^{n}$.

\end{corollary}

\begin{proof} By \eqref{eq:eigenvalue-eigenvector-pairs} and \eqref{eq:left-eigenvector},
for each $(i,j)\in\mathcal{I}$, we obtain $A^{\mathrm{T}}v_{i,j}v_{i,j}^{*}+v_{i,j}v_{i,j}^{*}A=\lambda_{i}v_{i,j}v_{i,j}^{*}+\overline{\lambda_{i}}v_{i,j}v_{i,j}^{*}=2\mathrm{Re}(\lambda_{i})v_{i,j}v_{i,j}^{*}.$
Since $v_{i,j}v_{i,j}^{*}\in\mathbb{C}^{n\times n}$ is a nonnegative-definite
Hermitian matrix, we have \eqref{eq:A-cond-1} with $R=v_{i,j}v_{i,j}^{*}$
and $\phi_{\mathrm{L}}=\mathrm{Re}(\lambda_{i})$. Moreover, by the
definition of $\mathcal{I}$, we have $\mathrm{tr}(D^{\mathrm{T}}v_{i,j}v_{i,j}^{*}D)=v_{i,j}^{*}DD^{\mathrm{T}}v_{i,j}>0$
for $(i,j)\in\mathcal{I}$. Thus, \eqref{eq:nonzero-noise-1} holds
with $R=v_{i,j}v_{i,j}^{*}$. Now, since \eqref{eq:A-cond-1} and
\eqref{eq:nonzero-noise-1} both hold for each $(i,j)\in\mathcal{I}$,
it follows from Theorem~\ref{TheoremContinuousTime} by setting $\beta_{\mathrm{U}}=\beta_{\mathrm{U}}^{i,j}$
that under control policies satisfying \eqref{eq:average-u-bound-1}
with $\hat{u}<\varphi_{i,j}$, the second moment of the state diverges.
Finally, \eqref{eq:ubarcond2} implies that there exists $(\tilde{i},\tilde{j})\in\mathcal{I}$
such that $\hat{u}<\varphi_{\tilde{i},\tilde{j}}$, implying divergence.
\end{proof}

\section{Numerical Examples}

\label{sec:Numerical-Example}

\begin{inpaper}

We consider two example systems and obtain ranges of $\hat{u}$ such
that these systems cannot be stabilized under the constraint \eqref{eq:average-u-bound-1}
when $\hat{u}$ is chosen from those ranges. 

\subsubsection*{Example 1}

Consider \eqref{eq:continuous-time-system} and \eqref{eq:Psi-C-def}
with 
\begin{align}
 & A=\left[\begin{array}{rr}
0 & 1\\
-1 & a
\end{array}\right],\,\,B=\left[\begin{array}{c}
0\\
1
\end{array}\right],\,\,D=\left[\begin{array}{c}
1\\
1
\end{array}\right],\nonumber \\
 & \ell_{1}=2,\,\,C_{1}=\left[\begin{array}{cc}
0 & 0\\
c_{1} & 0
\end{array}\right],\,\,C_{2}=\left[\begin{array}{cc}
0 & 0\\
0 & c_{2}
\end{array}\right],
\end{align}
where $a,c_{1},c_{2}\in\mathbb{R}$ are scalar coefficients. The case
with $a\geq0$ corresponds to the linearized dynamics of the forced
Van der Pol oscillator (see Section~5.3.1 of \cite{lakshmanan2012nonlinear}). 

For $a=-0.5$, $A$ is Hurwitz-stable; however, under multiplicative
noise with $c_{1}=c_{2}=2$, Theorem~3.1 indicates that the system
is instabilizable for $\hat{u}\in[0,7.4)$. For $a=1.5$, $c_{1}=c_{2}=2$,
the instabilizability range is $\hat{u}\in[0,10.8)$. Further, for
$a=1.5$, $c_{1}=c_{2}=0$ (i.e., no multiplicative noise), Theorem~3.1
provides the instabilizability range $\hat{u}\in[0,1)$ and Corollary~\ref{Corollary-Eigen}
provides a smaller range $\hat{u}\in[0,0.75)$. 

\subsubsection*{Example 2}

Consider \eqref{eq:continuous-time-system} and \eqref{eq:Psi-C-def}
with 
\begin{align*}
A & =\left[\begin{array}{cccc}
0 & 1 & 0 & 0\\
3\zeta^{2} & 0 & 0 & 2\zeta\\
0 & 0 & 0 & 1\\
0 & -2\zeta & 0 & 0
\end{array}\right],\,\,B=D=\left[\begin{array}{cc}
0 & 0\\
1 & 0\\
0 & 0\\
0 & 1
\end{array}\right],
\end{align*}
 $\ell_{1}=1$, $C_{1}=I_{4}$, where $\zeta$ is a scalar coefficient.
This system is a noisy version of the uncoupled, linearized, and normalized
dynamics that describes a satellite's motion in the equatorial plane,
as provided in \cite{fortmann1977introduction}. We check feasibility
of the linear matrix inequalities discussed in Remark~\ref{RemarkNumerical}
to assess the conditions of Theorem~\ref{TheoremContinuousTime}.
For $\zeta=1$ and $\hat{u}\in[0,1.1)$, the conditions of Theorem
3.1 hold, and thus, the system is instabilizable under the control
constraint \eqref{eq:average-u-bound-1}. 


\end{inpaper}

\begin{inarxiv} 

In this section, we illustrate our results on two example continuous-time
stochastic systems. 

\subsubsection*{Example 1 }

Consider the continuous-time stochastic system described by \eqref{eq:continuous-time-system}
and \eqref{eq:Psi-C-def} with {\fontsize{9.0}{11.7}\selectfont
\begin{align}
 & A=\left[\begin{array}{rr}
0 & 1\\
-1 & a
\end{array}\right],\,\,B=\left[\begin{array}{c}
0\\
1
\end{array}\right],\,\,D=d\left[\begin{array}{c}
1\\
1
\end{array}\right],\label{eq:example-system-def}\\
 & \ell_{1}=2,\,\,C_{1}=\left[\begin{array}{cc}
0 & 0\\
c_{1} & 0
\end{array}\right],\,\,C_{2}=\left[\begin{array}{cc}
0 & 0\\
0 & c_{2}
\end{array}\right],
\end{align}
}where $a,c_{1},c_{2},d\in\mathbb{R}$ are scalar coefficients. The
case with $a\geq0$ corresponds to the linearized dynamics of the
forced Van der Pol oscillator (see Section~5.3.1 of \cite{lakshmanan2012nonlinear}). 

In each row of Table~\ref{TableResults}, we consider a different
setting for $a,c_{1},c_{2},d\in\mathbb{R}$. Our goal is to obtain
ranges of $\hat{u}$ such that the system cannot be stabilized with
constrained control policies satisfying \eqref{eq:average-u-bound-1}
with $\hat{u}$ chosen in the given range. To obtain the range for
each parameter setting, we apply Theorem~\ref{TheoremContinuousTime}.
In addition, when there is no multiplicative noise (i.e., $c_{1}=c_{2}=0$),
we also apply Corollary~\ref{Corollary-Eigen}. 

\begin{table}
\centering

{\centering \renewcommand{\arraystretch}{0.8} \fontsize{7}{11.52}\selectfont 
\begin{tabular}{|c|r|c|c|c||l|l|}
\hline 
Setting \# & $a$ & $c_{1}$ & $c_{2}$ & $d$ & Applied Result & Range of $\hat{u}$\tabularnewline
\hline 
\hline 
1 & $-0.5$ & $2$ & $2$ & $1$ & Theorem~\ref{TheoremContinuousTime} & $[0,7.4)$\tabularnewline
\hline 
2 & $1.5$ & $2$ & $2$ & $1$ & Theorem~\ref{TheoremContinuousTime} & $[0,10.8)$\tabularnewline
\hline 
3 & $1.5$ & $2$ & $0$ & $1$ & Theorem~\ref{TheoremContinuousTime} & $[0,8.4)$\tabularnewline
\hline 
4 & $1.5$ & $0$ & $2$ & $1$ & Theorem~\ref{TheoremContinuousTime} & $[0,4.5)$\tabularnewline
\hline 
\multirow{2}{*}{5} & \multirow{2}{*}{$1.5$} & \multirow{2}{*}{$0$} & \multirow{2}{*}{$0$} & \multirow{2}{*}{$1$} & Theorem~\ref{TheoremContinuousTime} & $[0,1)$\tabularnewline
\cline{6-7} \cline{7-7} 
 &  &  &  &  & Corollary~\ref{Corollary-Eigen} & $[0,0.75)$ \tabularnewline
\hline 
\multirow{2}{*}{6} & \multirow{2}{*}{$1.5$} & \multirow{2}{*}{$0$} & \multirow{2}{*}{$0$} & \multirow{2}{*}{$2$} & Theorem~\ref{TheoremContinuousTime} & $[0,4.1)$\tabularnewline
\cline{6-7} \cline{7-7} 
 &  &  &  &  & Corollary~\ref{Corollary-Eigen} & $[0,3)$ \tabularnewline
\hline 
\end{tabular}}

\caption{Ranges of $\hat{u}$ for which stabilization is impossible. }
 \label{TableResults} 
\end{table}

Setting~1 in Table~\ref{TableResults} represents the case where
$A$ is a Hurwitz-stable matrix. Notice that when $A$ is Hurwitz-stable,
without multiplicative noise, the second moment of the state of the
uncontrolled system would remain bounded. However, as Table~\ref{TableResults}
indicates, under multiplicative noise (with parameters $c_{1}=c_{2}=2$),
the system is unstable under any constrained control that satisfies
\eqref{eq:average-u-bound-1} with $\hat{u}\in[0,7.4)$. 

Settings 2--6 in Table~\ref{TableResults} represent different scenarios
where $A$ is strictly unstable. In each of those settings, different
noise parameters are considered. The main observation is that systems
with increased noise levels are harder to stabilize with constrained
controllers. 

Settings~5 and 6 in Table~\ref{TableResults} correspond to the
cases where there is no multiplicative noise. In those cases Corollary~\ref{Corollary-Eigen}
can be applied. Notice that Corollary~\ref{Corollary-Eigen} provides
smaller ranges for $\hat{u}$ compared to Theorem~\ref{TheoremContinuousTime}.
This shows that Corollary~\ref{Corollary-Eigen} is conservative.
We note that the main advantage of Corollary~\ref{Corollary-Eigen}
is that its conditions can be checked faster than those of Theorem~\ref{TheoremContinuousTime}. 

For checking conditions of Corollary~\ref{Corollary-Eigen}, we can
speedily compute $\max_{(i,j)\in\mathcal{I}}\varphi_{i,j}$ in \eqref{eq:ubarcond2}.
In particular, the computation yields the analytical expression 
\begin{align}
\max_{(i,j)\in\mathcal{I}}\varphi_{i,j} & =\begin{cases}
d^{2}(2-a)a, & a\in[0,2),\\
4d^{2}(\frac{1}{2}a-1), & a\geq2,
\end{cases}\label{eq:max_psi_eq}
\end{align}
which is also shown in Fig.~\ref{Figure3d}. Given $a$ and $d$,
if $\hat{u}<\max_{(i,j)\in\mathcal{I}}\varphi_{i,j}$ (i.e., the value
is below the surface in Fig.~\ref{Figure3d}), then Corollary~\ref{Corollary-Eigen}
implies that stabilization of \eqref{eq:continuous-time-system} is
impossible with a control policy that satisfies \eqref{eq:average-u-bound-1}
with that particular $\hat{u}$. Notice that $\max_{(i,j)\in\mathcal{I}}\varphi_{i,j}$
is a quadratic function of $d$, and thus, for larger values of $d$,
the value of $\max_{(i,j)\in\mathcal{I}}\varphi_{i,j}$ becomes larger.
This result is intuitive in the sense that stabilization becomes harder
under stronger noise. On the other hand, $\max_{(i,j)\in\mathcal{I}}\varphi_{i,j}$
depends on $a$ in a nonlinear nonmonotonic way. It follows from \eqref{eq:max_psi_eq}
that for $a\in[0,2)$, the maximum of $\max_{(i,j)\in\mathcal{I}}\varphi_{i,j}$
is achieved when $a=1$. For $a\geq2$, $\max_{(i,j)\in\mathcal{I}}\varphi_{i,j}$
increases as $a$ increases. 

\begin{figure}[t]
\centering \includegraphics[width=0.85\columnwidth]{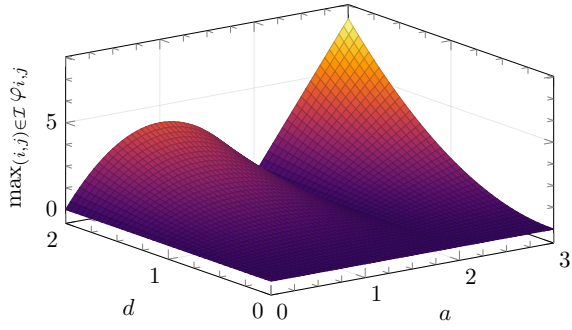}

\caption{The value of $\max_{(i,j)\in\mathcal{I}}\varphi_{i,j}$ in \eqref{eq:ubarcond2}.
If \eqref{eq:average-u-bound-1} holds with $\hat{u}<\max_{(i,j)\in\mathcal{I}}\varphi_{i,j}$,
then stabilization is impossible.}
 \label{Figure3d} 
\end{figure}

\subsubsection*{Example 2 }

Consider \eqref{eq:continuous-time-system} and \eqref{eq:Psi-C-def}
with 
\begin{align*}
A & =\left[\begin{array}{cccc}
0 & 1 & 0 & 0\\
3\zeta^{2} & 0 & 0 & 2\zeta\\
0 & 0 & 0 & 1\\
0 & -2\zeta & 0 & 0
\end{array}\right],\,\,B=D=\left[\begin{array}{cc}
0 & 0\\
1 & 0\\
0 & 0\\
0 & 1
\end{array}\right],\\
\ell_{1} & =1,\,\,C_{1}=I_{4}.
\end{align*}
 This system is a noisy version of the uncoupled, linearized, and
normalized dynamics that describes a satellite's motion in the equatorial
plane, as provided in \cite{fortmann1977introduction}. The scalar
$\zeta$ is the angular velocity of the equatorial orbit along which
the system is linearized and the control input $u(t)\in\mathbb{R}^{2}$
is the vector of thrusts applied to the satellite in the equatorial
plane. We consider the control input constraint \eqref{eq:average-u-bound-1}
with the average second moment bound~$\hat{u}$. 

We check feasibility of the linear matrix inequalities discussed in
Remark~\ref{RemarkNumerical} to assess the conditions of Theorem~\ref{TheoremContinuousTime}
for different values of $\zeta$ and $\hat{u}$. When $\zeta=0.1$,
the conditions of Theorem 3.1 hold for $\hat{u}\in[0,1.7)$. Thus
the system is instabilizable under the control constraint \eqref{eq:average-u-bound-1}
with those values of $\hat{u}$. On the other hand, with $\zeta=1$,
the corresponding instabilizability range  is obtained as $\hat{u}\in[0,1.1)$. 

\end{inarxiv}

\section{Conclusion}

\label{sec:Conclusion}

We have investigated the constrained control problem for linear stochastic
systems with additive and multiplicative noise terms. We have shown
that in certain scenarios, stabilization is impossible to achieve
with control policies that have bounded time-averaged second moments.
In particular, we have obtained conditions, under which the second
moment of the system state diverges regardless of the controller design
and regardless of the initial state. Moreover, we have showed the
tightness of our results for scalar systems and provided extensions
for partially-constrained control policies and additive-only noise
settings. 

\bibliographystyle{ieeetr}
\bibliography{references}

\end{document}